\documentclass[letterpaper, 10 pt, conference]{ieeeconf}
\IEEEoverridecommandlockouts
\usepackage{amsmath}
\usepackage{amssymb}
\usepackage{graphicx}
\usepackage[colorlinks=true, allcolors=blue]{hyperref}
\usepackage[permil]{overpic}
\usepackage{caption}
\usepackage{subcaption}
\newtheorem{assumption}{Assumption}
\newtheorem{theorem}{Theorem}

\newtheorem{lemma}{Lemma}

\everymath{\displaystyle}
\graphicspath{ {Figures/} }

\title{\LARGE \bf Real-Time Sensor-Based Feedback Control for Obstacle Avoidance in Unknown Environments}
\author{Lyes Smaili and Soulaimane Berkane, \IEEEmembership{Senior Member, IEEE}
\thanks{This research work is supported in part by NSERC-DG RGPIN-2020-04759 and Fonds de recherche du Qu\'ebec (FRQ).}
\thanks{The authors are with the Department of Computer Science and Engineering, University of Quebec in Outaouais, Gatineau, QC, Canada. {\tt\small smal01@uqo.ca}, {\tt\small Soulaimane.Berkane@uqo.ca}. S. Berkane is also with the Department of Electrical Engineering, Lakehead University, ON, Canada.}
}

\date{August 2023}

\begin{document}

\maketitle
\thispagestyle{empty}
\pagestyle{empty}

\begin{abstract}
 We revisit the Safety Velocity Cones (SVCs) obstacle avoidance approach for real-time autonomous navigation in an unknown $n$-dimensional environment. We propose a locally Lipschitz continuous implementation of the SVC controller using the distance-to-the-obstacle function and its gradient. We then show that the proposed implementation guarantees safe navigation in generic environments and almost globally asymptotic stability (AGAS) of the desired destination when the workspace contains strongly convex obstacles.  The proposed computationally efficient control algorithm can be implemented onboard vehicles equipped with limited range sensors ({\it e.g.,} LiDAR, depth camera), allowing the controller to be locally evaluated without requiring prior knowledge of the environment.  
\end{abstract}

\section{Introduction}
\subsection{Motivation and Prior Works}
Given its valuable use in diverse applications, the design of autonomous navigation systems is a highly addressed topic in robotics. For a robot to navigate safely in an environment cluttered with obstacles, it is essential to devise an effective control strategy capable of resolving the obstacle avoidance problem. One of the simplest and computationally efficient techniques is the artificial potential fields approach \cite{khabib1986} that allows real-time obstacle avoidance. However, even for simple environments, the constructed potential-field may admit many local minima \cite{Koditschek1987}. To solve this issue, different global approaches have been proposed. The navigation functions approach \cite{KODITSCHEK1990}, when applied to topologically simple environments like Euclidean sphere worlds \cite{KoditschekRimon1992}, solves the problem of local minima through an appropriate parameter tuning, and ensures almost global stability of the target location. The navigation functions approach can be extended to more generic environments by applying some diffeomorphic mappings \cite{KoditschekRimon1992}, \cite{LiCailiTanner2019}. Other methods such as the navigation transform \cite{Loizou2017} and the prescribed performance control \cite{Vrohidis2018} do not necessitate any parameter tuning to eliminate the local minima. A feedback control strategy has been proposed in \cite{cheniouni2023safe} to ensure safety while avoiding obstacles through the shortest path. Hybrid feedback has been used in works such as \cite{sanfelice2006robust,berkane2019hybrid,casau2019hybrid,berkane2021obstacle} to remove the hassle of undesired equilibria and ensure global asymptotic stability. However, global methods require complete prior knowledge of the environment.

On the other hand, reactive methods allow real-time navigation in unknown environments, an important autonomy feature required for contemporary applications of autonomous robotics. The navigation function-based methods were extended to navigate in unknown and topologically simple environments in \cite{Lionis2007} and \cite{Filippidis2011}. The sensor-based method in \cite{arslan2019sensor} makes use of separating hyperplanes to identify the local free space of the robot, and it guarantees almost global asymptotic stability in environments filled with separated and strongly convex obstacles. Navigation through safety velocity cones \cite{berkane2021Navigation} makes use of Nagumo's invariance theorem \cite{Nagumo} by projecting the nominal controls (nominal velocities) onto the Bouligand's tangent cones \cite{bouligand1932introduction} to ensure safety for general spaces. However, even for Euclidean sphere worlds, the discontinuous approach in \cite{berkane2021Navigation} might result in saddle points with a basin of attraction that is not of measure zero. The hybrid feedback control approaches in \cite{sawant2023hybrid} and \cite{Sawant2023} allow safe navigation in unknown two-dimensional environments with convex and non-convex obstacles, respectively. 

\subsection{Contributions of the Paper}
In this present paper, we revisit our previously proposed approach in \cite{berkane2021Navigation} which uses safety velocity cones to guarantee safety in arbitrary environments. In this work, we focus our attention on providing convergence guarantees when navigating an $n$-dimensional unknown environment filled with strongly convex obstacles; a similar setting to \cite{arslan2019sensor}.  We summarize the contributions of this paper as follows:

\begin{enumerate}
    \item By construction, following \cite{berkane2021Navigation}, our controller guarantees safety and progress towards the target in very generic environments (not only convex). This is a very appealing feature compared to most of the proposed algorithms that are tailored usually to the specific setting, \textit{e.g.,} \cite{arslan2019sensor,sawant2023hybrid,cheniouni2023safe}.
    \item By considering a smoothed version of the discontinuous controller in \cite{berkane2021Navigation}, and for sufficiently convex and smooth obstacles, we prove that the closed-loop dynamical system admits a unique solution that converges safely to the exponentially stable desired destination from almost all initial conditions in the free space. 
    \item Our proposed controller can be evaluated without any prior knowledge about the environment. The controller is computationally efficient and suitable for real-time implementation as it requires only measurements of the range and bearing to the nearest obstacle (obtained, for example, using range scanners). 
\end{enumerate}
    
\subsection{Organization of the Paper}
This paper is organized as follows. In Section \ref{section:GeneralWorkspaces} we define the workspace in general terms, as well as some assumptions on its topology. In Section \ref{section:DistanceBasedSmoothController} we formulate the problem, and we present the smooth controller, with the related notions such as safety and convergence. In Section \ref{section:GeneralWorkspaces} we present the convex sphere worlds, and we prove almost global asymptotic stability when having strongly convex obstacles. In Section \ref{section:NumericalSimulation} we demonstrate the effectiveness of our navigation algorithm via numerical simulations in 2D and 3D environments. In Section \ref{section:conclusion} we conclude with a summary of our work, and we discuss related future works.

\subsection{Notation}
We denote by $\mathbb{R}$, $\mathbb R_{>0}$ and $\mathbb N$, respectively, the set of reals, positive reals and natural numbers. We denote by $\mathbb R^n$ the $n$-dimensional Euclidean space and by $\mathbb S^{n-1}$ the $(n-1)$-dimensional unit sphere embedded in $\mathbb R^n$, with $n\in\mathbb N$. We denote the Euclidean norm of a vector $x\in\mathbb R^n$ by $||x||$. For a subset $\mathcal{A}\subset\mathbb R^n$, we denote by $\textbf{int}(\mathcal{A})$, $\partial\mathcal{A}$, $\overline{\mathcal{A}}$ and $\complement\mathcal{A}$, respectively, its topological interior, boundary, closure and complement in $\mathbb R^n$. We denote the Euclidean ball of radius $r>0$ centered at $x$ by $\mathcal{B}(x,r):=\{y\in\mathbb R^n: ||x-y||<r\}$. The distance from a point $x\in\mathbb R^n$ to a closed set $\mathcal{A}\subset\mathbb R^n$ is given by $\textbf{d}_\mathcal{A}(x):=\inf _{y\in\mathcal{A}}||y - x||$.  For two sets $\mathcal{A},\mathcal{B}\subset\mathbb{R}^n$, the distance between them is given by
$\mathbf{d}_{\mathcal{A},\mathcal{B}}:= \inf _{x\in\mathcal{A},y\in\mathcal{B}}||y - x||$. The projection of $x\in\mathbb{R}^n$ onto $\mathcal{A}\subset\mathbb R^n$ is given by $\textbf{P}_\mathcal{A}(x):=\{y\in\overline{\mathcal{A}}: ||y-x||=\textbf{d}_\mathcal{A}(x)\}$. If the projection $\textbf{P}_{\partial\mathcal{A}}(x)$ is unique for some $x\in\textbf{int}(\mathcal{A})$, the inward normal vector of the set $\mathcal{A}$ at $\textbf{P}_{\partial\mathcal{A}}(x)$ is given by the gradient of the distance function $\textbf{d}_{\complement\mathcal{A}}(x)$ (see \cite[theorem 3.3, chap 6]{delfour2011shapes}):
\begin{equation}
    \nabla\textbf{d}_{\complement\mathcal{A}}(x):=\frac{x-\textbf{P}_{\partial\mathcal{A}}(x)}{||x-\textbf{P}_{\partial\mathcal{A}}(x)||},\hspace{5mm}\forall x\in\textbf{int}(\mathcal{A}).
\end{equation}
We define the oriented distance function as $\mathbf{b}_\mathcal{A}(x):= \textbf{d}_\mathcal{A}(x)-\textbf{d}_{\complement\mathcal{A}}(x)$, see \cite[definition 2.1, chap 7]{delfour2011shapes}. The gradient of the oriented distance function, for all $x\in \mathbb R^n\setminus\{\partial\mathcal{A}\}$ and for all $x$ such that $\textbf{P}_{\partial\mathcal{A}}(x)$ is unique, is given by \cite[theorem 3.1, chap. 7]{delfour2011shapes}

\begin{equation}
    \nabla\textbf{b}_{\complement\mathcal{A}}(x):=\left\{
    \begin{array}{cc}
      -\frac{x-\textbf{P}_{\partial\mathcal{A}}(x)}{||x-\textbf{P}_{\partial\mathcal{A}}(x)||},   &  x\in\textbf{int}(\complement\mathcal{A}),\\
      
      \frac{x-\textbf{P}_{\partial\mathcal{A}}(x)}{||x-\textbf{P}_{\partial\mathcal{A}}(x)||},   &  x\in\textbf{int}(\mathcal{A}).   
    \end{array}
    \right.
\end{equation}
We denote by $\mathbf{Sk}(\mathcal{A})$, the skeleton of $\mathcal{A}$, a set of all points of $\mathbb R^n$ whose projection onto $\overline{\mathcal{A}}$ is not unique, defined as
\begin{equation}
    \mathbf{Sk}(\mathcal{A}):=\{x\in\mathbb R^n:\mathbf{card}(\mathbf{P}_{\overline{\mathcal{A}}}(x))>1\}.
\end{equation}

    For a non-empty set $\mathcal{A}$, the reach of $\mathcal{A}$ at $x\in\overline{\mathcal{A}}$ is defined as
    \begin{align}
        &\mathbf{reach}(\mathcal{A},x):=\nonumber\\
        &\left\{\begin{array}{ll}
          0,   &  x\in\partial\overline{\mathcal{A}}\cap\overline{\mathbf{Sk}(\mathcal{A})},\\
          \sup\{r>0:\mathbf{Sk}(\mathcal{A})\cap\mathcal{B}(x,r)=\emptyset \},   & \text{otherwise.}
        \end{array}\right.
    \end{align}
    The reach of the set $\mathcal{A}$ is given by \cite[Definition 6.1, chap. 6]{delfour2011shapes}
    \begin{equation}
        \mathbf{reach}(\mathcal{A}):=\inf_{x\in\mathcal{A}}\{\mathbf{reach}(\mathcal{A},x)\}.
    \end{equation}
    The set $\mathcal{A}$ has {\it positive reach} if $\mathbf{reach}(\mathcal{A})>0$.

Let $f:\mathbb R^n\to\mathbb R^m$ be a vector-valued function, where $f(x)=[f_1(x),f_2(x),\cdots,f_m(x)]^\top$. The Jacobian matrix of the function $f$ with respect to $x=[x_1,x_2,\cdots,x_n]^\top$ is an $m\times n$ matrix defined as
\begin{equation}
    \textbf{J}_f(x):=\begin{bmatrix}
        \frac{\partial f_1}{\partial x_1}& \frac{\partial f_1}{\partial x_2}&\cdots&\frac{\partial f_1}{\partial x_n}\\
        \frac{\partial f_2}{\partial x_1}& \frac{\partial f_2}{\partial x_2}&\cdots&\frac{\partial f_2}{\partial x_n}\\
        \vdots& \vdots&\ddots&\vdots\\
        \frac{\partial f_m}{\partial x_1}& \frac{\partial f_m}{\partial x_2}&\cdots&\frac{\partial f_m}{\partial x_n}
    \end{bmatrix}.
\end{equation}
We denote by $\textbf{I}_n$ the $n\times n$ identity matrix.
Finally, we denote by $\textbf{T}_\mathcal{A}(x)$ the tangent cone to $\mathcal{A}$ at a point $x\in\mathbb R^n$ \cite{bouligand1932introduction}, which is given by
\begin{equation}
    \mathbf{T}_\mathcal{A}(x)=\big\{z\in\mathbb R^n:\lim_{\tau\to 0^+}\frac{\textbf{d}_{\mathcal{A}}(x+\tau z)}{\tau}=0\big\}.
\end{equation}

\section{Problem Formulation}\label{section:GeneralWorkspaces}
Let $\mathcal{W}$ be a closed subset of the $n$-dimensional Euclidean space $\mathbb{R}^n$, which bounds the workspace. Consider $M$ smaller open sets $\mathcal{O}_i$, $i=1,...,M$ in $\mathbb{R}^n$, strictly contained in the interior of $\mathcal{W}$, that describes the obstacles. We denote by $\mathcal{O}_0:=\partial\mathcal{W}$ the boundary obstacle and we let $\mathbb{M}:=\{0,1,\cdots,M\}$. The free space can be described by the closed set $\mathcal{X}$, which is given by
\begin{equation}
    \mathcal{X}:=\mathcal{W}\setminus \bigcup_{i=1}^{M}\mathcal{O}_i.
\end{equation}

The complement of the free space set, \textit{i.e.,} the set $\complement\mathcal{X}$, represents the obstacle region. Let $\mathbf{reach}(\mathcal{X})$ be the reach of the set $\mathcal{X}$ and consider the following assumption:
\begin{assumption}[Positive Reach set]\label{assumption:positive reach}
   The set $\mathcal{X}$ has a positive reach, {\it i.e.},
   \begin{equation}
       \mathbf{reach}(\mathcal{X})>0.
   \end{equation}
\end{assumption}

In other words, there exists a positive real $h>0$ such that any point $x\in\mathcal{X}$, with $\mathbf{d}_{\complement\mathcal{X}}\left(x\right)<h$, has a unique projection $\mathbf{P}_{\partial\mathcal{X}}\left(x\right)$ \cite[Theorem
6.3, Chap. 6]{delfour2011shapes}. If Assumption \ref{assumption:positive reach} holds, the following inequality is true
\begin{equation}
        \begin{array}{cc}
            \mathbf{d}_{\mathcal{O}_i,\mathcal{O}_j}>2h,   & \forall i,j\in\mathbb{M} \textrm{ with }i\ne j,
        \end{array}
\end{equation}
which means that the obstacles are separated by, at least, a distance of $2h$.

We consider a ball-shaped robot centred at $x\in \mathbb{R}^n$ with radius $R>0$, and we define the {\it practical free space} as follows\footnote{The controller can be equivalently formulated using the regular distance function as in \cite{berkane2021Navigation}. However, the use of the oriented distance function proves important to establish the proofs of the main result.}:
\begin{equation}\label{eq:practicalfreespace}
    \mathcal{X}_{\epsilon}:=\{x\in\mathbb{R}^n: \mathbf{b}_{\complement\mathcal{X}}(x)\ge\epsilon\}\subset\mathcal{X},
\end{equation}
where $\epsilon$ is a positive safety margin. For feasibility, the choice of the safety margin should satisfy the condition
\begin{align}\label{condition:1}
    0<R<\epsilon<h.
\end{align}
This condition guarantees for the robot a safe distance ($\epsilon>R$) from the obstacles, and if Assumption \ref{assumption:positive reach} holds, then the obstacles are separated enough for the robot to move freely in between.

\begin{figure}[t]
    \centering
    \includegraphics[width=0.8\columnwidth]{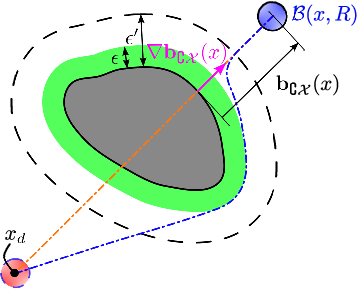}
    \caption{Obstacle avoidance of a ball-shaped robot (blue). The (green) region around the obstacle (gray) is a dilation for the latest by the parameter $\epsilon$ while the (black) dashed line is a dilation by the parameter $\epsilon'$. The (orange) trajectory results from the application of the nominal controller which brings the robot's position to the desired goal $x_d$ (red) in the absence of the obstacle. The (blue) trajectory results from applying the smooth controller which brings the robot's position to the desired goal while avoiding the dilated obstacle.
    } 
    \label{fig:figurePath}
\end{figure}
We consider a robot operating in the free space $\mathcal{X}$ and restricted to stay in the practical free space $\mathcal{X}_{\epsilon}$ defined by (\ref{eq:practicalfreespace}). We assume the following first-order robot dynamics
\begin{equation}\label{eq:dynamicalSystem}
    \dot{x}=u,
\end{equation}
\noindent where $x\in\mathbb{R}^n$ is the position of the robot's center, and $u\in\mathbb R^n$ is the control input (velocity).

The reactive navigation problem consists in finding a Lipschitz continuous controller $u=\kappa(x,x_d,\mathcal{X}_\epsilon)$, $\kappa:\mathbb{R}^n\times\mathbb{R}^n\times\mathbb{R}^n\to\mathbb{R}^n$, such that, for the closed-loop system 
\begin{equation}\label{eq:closedLoopsystem}
    \dot{x}=\kappa(x,x_d,\mathcal{X}_\epsilon),
\end{equation}
the practical free space $\mathcal{X}_{\epsilon}$ is positively invariant, and the robot's position is asymptotically stabilized at a given desired location $x_d\in\mathbf{int}\left(\mathcal{X}_{\epsilon}\right)$. Also, we must guarantee that the control law can be computed in real-time, using only locally known information about the free space $\mathcal{X}$. For the sake of simplicity, we are going to write $\kappa(x)$ instead of $\kappa(x,x_d,\mathcal{X}_\epsilon)$.

\section{Distance-based smooth controller}\label{section:DistanceBasedSmoothController}
\subsection{Feedback Control Design}
A solution to the reactive navigation problem has been derived in \cite{berkane2021Navigation} and consists of solving the following \textit{nearest
point problem}:
\begin{equation}\label{eq:optimization}
\min_u ||u-\kappa_0(x)||^2 \hspace{2mm}\text{subject to} \hspace{2mm}u\in\textbf{T}_{\mathcal{X}_{\epsilon}}(x), \forall x\in\mathcal{X}_\epsilon,
\end{equation}
where $\kappa_0(x)$ is the nominal control law that stabilizes the robot to the target location in the absence of obstacles. 
The motivation behind this optimization problem is to ensure the conditions of Nagumo's theorem for invariance while minimally deviating from the nominal controller (minimally invasive control). In order to satisfy the necessary and sufficient conditions of Nagumo's theorem, the control law $u$ (vector field) must be constrained to the tangent cone (coined {\it safety velocity cone} (SVC) in \cite{berkane2021Navigation}).

The solution to this optimization problem is equivalent to finding the projection of the nominal control $\kappa_0(x)$ onto the tangent cone $\textbf{T}_{\mathcal{X}_{\epsilon}}(x)$. The projection operator is denoted $\mathbf{P}(\kappa_0(x),\mathbf{T}_{\mathcal{X}_{\epsilon}}(x))$. We recognize two cases depending on the position $x$ of the robot. When $x\in\mathbf{int}(\mathcal{X}_\epsilon)$, the tangent cone set is given by $\mathbf{T}_{\mathcal{X}_{\epsilon}}(x)\equiv\mathbb R^n$, and when $x\in\partial\mathcal{X}_\epsilon$, the tangent cone set depends on the shape of the boundary. For arbitrary free spaces, the projection $\mathbf{P}(\kappa_0(x),\mathbf{T}_{\mathcal{X}_{\epsilon}}(x))$ need not to be unique. \\
According to \cite[Thm. 7.1, Chap.7]{delfour2011shapes}, which states that, for any non-empty set $\complement\mathcal{X}$ with positive reach $h$, the dilated set $\complement\mathcal{X}_\epsilon$ is a set of class $\mathcal{C}^{1,1}$, where $\mathcal{X}_\epsilon$ is given by (\ref{eq:practicalfreespace}). Hence, the boundary $\partial\mathcal{X}_\epsilon$ is a $\mathcal{C}^{1,1}$-submanifold of dimension $(n-1)$. Therefore, the tangent cone at any $x\in\partial\mathcal{X}_{\epsilon}$ is given by the half-space
\begin{equation}
    \mathbf{T}_{\mathcal{X}_{\epsilon}}(x)=\{z\in\mathbb{R}^n:v\left(x\right)^\top z\le 0\}, \forall x\in \partial\mathcal{X}_{\epsilon},
\end{equation}
\noindent where $v\left(x\right)$ is the outward normal unit vector associated to each $x\in\partial\mathcal{X}_{\epsilon}$. A half-space is a convex set. Therefore, any vector $\kappa_0\left(x\right)\in\mathbb{R}^n$, defined at $x\in\partial\mathcal{X}_{\epsilon}$, has a unique projection $\mathbf{P}\left(\kappa_0\left(x\right),\mathbf{T}_{\mathcal{X}_{\epsilon}}(x)\right)$ onto the tangent space $\mathbf{T}_{\mathcal{X}_{\epsilon}}(x)$. When $v\left(x\right)^\top\kappa_0\left(x\right)>0$, the projection reduces to the orthogonal projection onto the hyperplane $v\left(x\right)^\top z=0$, which is given by \cite[Chap. 5]{Meyer2000}
\begin{equation}
    \Pi\left(v\left(x\right)\right)\kappa_0\left(x\right):=\left(\mathbf{I}_n-v\left(x\right)v\left(x\right)^\top\right)\kappa_0\left(x\right).
\end{equation}
The resulting control law that solves \eqref{eq:optimization} is given by
\begin{equation}\label{eq:disControl1}
    u= \kappa(x) =\left\{
        \begin{array}{ l l }
    \kappa_0\left(x\right), & x\in\textbf{int}(\mathcal{X}_\epsilon)\hspace{2mm} \text{or}\\
    {}&  v(x)^\top\kappa_0(x)\le0,\\
    \Pi(x)\kappa_0(0), & x\in\partial\mathcal{X}_\epsilon \hspace{2mm}\text{and}\\
    {}&  v(x)^\top\kappa_0(x)\ge0,
  \end{array}
    \right.
\end{equation}
which is a discontinuous vector field at the boundary $\partial\mathcal{X}_\epsilon$ of the practical free space. The discontinuity appears when the nominal controller points outside of the practical free set $\mathcal{X}_\epsilon$. In this case, the nominal controller is projected onto the tangent cone set. Since $\mathcal{X}_\epsilon$ is a dilation of $\mathcal{X}$, then 
\begin{equation}
    \nabla\mathbf{b}_{\complement\mathcal{X}}(x)=\nabla\mathbf{b}_{\complement\mathcal{X}_\epsilon}(x),\hspace{1cm}\forall x\in\mathbf{int}(\mathcal{X}_\epsilon),
\end{equation}
where, $\nabla\mathbf{b}_{\complement\mathcal{X}_\epsilon}(x)$ is the inward normal vector of $\partial\mathcal{X}_\epsilon$ at $\mathbf{P}_{\partial\mathcal{X}_\epsilon}(x)$, $\forall x\in\mathbf{int}(\mathcal{X}_\epsilon)$. Therefore,
\begin{equation}
    \nabla\mathbf{b}_{\complement\mathcal{X}}(x)=-v(\mathbf{P}_{\partial\mathcal{X}_\epsilon}(x)),\hspace{1cm}\forall x\in\mathcal{X}_\epsilon.
\end{equation}

We can rewrite (\ref{eq:disControl1}) in terms of the oriented distance function from the robot position $x$ to the obstacle set $\complement\mathcal{X}$, and using the fact that $v(x)=-\nabla\mathbf{b}_{\complement\mathcal{X}}(x)$, $\forall x\in\partial\mathcal{X}_\epsilon$, as follows
\begin{equation}
  \setlength{\arraycolsep}{5pt}
  \kappa(x)  = \left\{ \begin{array}{ l l }
    \kappa_0\left(x\right), &\mathbf{b}_{\complement	\mathcal{X}}\left(x\right)>\epsilon \hspace{2mm}\text{or} \\
    {}&  \kappa_0\left(x\right)^\top\nabla\mathbf{b}_{\complement	\mathcal{X}}\left(x\right)\ge 0,\\
    \Pi\left(x\right)\kappa_0\left(x\right), &\mathbf{b}_{\complement	\mathcal{X}}\left(x\right)=\epsilon \hspace{2mm}\text{and} \\
    {}&  \kappa_0\left(x\right)^\top\nabla\mathbf{b}_{\complement	\mathcal{X}}\left(x\right)\le 0,\\
  \end{array} \right.
\end{equation}

To get rid of the discontinuity, we propose the following smoothed version inspired from \cite[Appendix A]{Praly1991}
\begin{equation}\label{eq:smoothControl}
  \setlength{\arraycolsep}{5pt}
  \kappa(x) = \left\{ \begin{array}{ l l }
    \kappa_0\left(x\right), &\mathbf{b}_{\complement	\mathcal{X}}\left(x\right)>\epsilon' \hspace{2mm}\text{or} \\
    {}&  \kappa_0\left(x\right)^\top\nabla\mathbf{b}_{\complement	\mathcal{X}}\left(x\right)\ge 0,\\
    \hat{\Pi}\left(x\right)\kappa_0\left(x\right), &\mathbf{b}_{\complement	\mathcal{X}}\left(x\right)\le\epsilon' \hspace{2mm}\text{and} \\
    {}&  \kappa_0\left(x\right)^\top\nabla\mathbf{b}_{\complement	\mathcal{X}}\left(x\right)\le 0,
  \end{array} \right.
\end{equation}

\noindent where $0<R<\epsilon<\epsilon'\le h$ and 
\begin{equation}
    \hat{\Pi}\left(x\right):=\mathbf{I}_n-\phi\left(x\right)\nabla\mathbf{b}_{\complement\mathcal{X}}(x)\nabla\mathbf{b}_{\complement\mathcal{X}}(x)^\top,
\end{equation}

\begin{equation}\label{eq:phi}
    \phi\left(x\right):=\min\left(1,\frac{\epsilon'-\mathbf{b}_{\complement\mathcal{X}}(x)}{\epsilon'-\epsilon}\right).
\end{equation}

The smoothness of $\kappa(x)$ depends on the class of the oriented distance function $\mathbf{b}_{\complement\mathcal{X}}$ which, also, depends on the class of the set $\mathcal{X}$. Therefore, we assume the following smoothness assumption for the free space. 
\begin{assumption}\label{assumption:smoothBoundaries}
    The free space $\mathcal{X}$ is a set of class $\mathcal{C}^{2,l}$, where $0\le l\le 1$.
\end{assumption}
We refer to \cite[Definition 3.1, Chap 2]{delfour2011shapes} to define sets of class $\mathcal{C}^{k,l}$, where $k\ge 1$ is an integer and $0\le l\le 1$ is a real. 
\begin{lemma}
    Consider the practical free space set $\mathcal{X}_\epsilon$. Under Assumption \ref{assumption:smoothBoundaries}, and assuming $\kappa_0$ is locally Lipschitz-continuous, the smoothed control $\kappa(x)$ given by (\ref{eq:smoothControl})-(\ref{eq:phi}) is locally Lipschitz-continuous.
\end{lemma}
\begin{proof}
    This proof is inspired from the proof of Lemma 103 presented in \cite[Appendix A]{Praly1991}. Firstly, we have that $\mathbf{b}_{\complement\mathcal{X}}(x)$, under Assumption \ref{assumption:smoothBoundaries}, is a twice continuously differentiable function. In fact, according to \cite[Theorem. 8.2, Chap. 7] {delfour2011shapes}, if the free space $\mathcal{X}$ is a $\mathcal{C}^{k,l}$-class set, then
        \begin{equation*}
            \forall x\in\partial\mathcal{X}, \exists \rho>0\hspace{2mm} \text{such that}\hspace{2mm} \mathbf{b}_{\complement\mathcal{X}}\in\mathcal{C}^{k,l}(\overline{\mathcal{B}(x,\rho)}).
        \end{equation*}

Let $\mathrm{Proj}(x,\kappa): \mathcal{X}_\epsilon\times\mathbb R^n\to\mathbb R^n$ be a projection map defined as follows:

\begin{equation}\label{eq:projection}
    \mathrm{Proj}(x,\kappa)= \left\{ \begin{array}{ l l }
    \kappa, &\mathbf{b}_{\complement	\mathcal{X}}\left(x\right)>\epsilon' \hspace{2mm}\text{or} \\
    {}&  \kappa^\top\nabla\mathbf{b}_{\complement	\mathcal{X}}\left(x\right)\ge 0,\\
    \hat{\Pi}\left(x\right)\kappa, &\mathbf{b}_{\complement	\mathcal{X}}\left(x\right)\le\epsilon' \hspace{2mm}\text{and} \\
    {}&  \kappa^\top\nabla\mathbf{b}_{\complement	\mathcal{X}}\left(x\right)\le 0,
  \end{array} \right.
\end{equation}
such that, when $\kappa=\kappa_0(x)$ one has $\mathrm{Proj}(x,\kappa_0(x))=\kappa(x)$. We denote by $\mathcal{S}$ the following open subset of  $\mathcal{X}_\epsilon\times\mathbb R^n$
        \begin{equation}
            \mathcal{S}=\{(x,\kappa): \mathbf{b}_{\complement	\mathcal{X}}\left(x\right)<\epsilon', \kappa^\top\nabla\mathbf{b}_{\complement	\mathcal{X}}\left(x\right)< 0\}.
        \end{equation}
        Then the function $\mathrm{Proj}(x,\kappa)$ is continuously differentiable at $\mathcal{S}$. The projection $\mathrm{Proj}(x,\kappa)$ tends to $\kappa$ as $\mathbf{b}_{\complement\mathcal{X}}(x)$ tends to $\epsilon'$ or as $ \kappa^\top\nabla\mathbf{b}_{\complement	\mathcal{X}}\left(x\right)$ tends to $0$. For any compact subset $\mathcal{C}$ of $\mathcal{S}$, there exists a constant $k_\mathcal{C}$ such that the Jacobian matrix:
        \begin{equation}\label{eq:jacobianKappaSmooth}
            \|\mathbf{J}_\mathrm{Proj}(x,\kappa) \|\le k_\mathcal{C},\hspace{2mm}\forall (x,\kappa)\in\mathcal{C}.
        \end{equation}
    Let $(x_a,\kappa_{a})$ and $(x_b,\kappa_{b})$ be two distinct points such that, for any $\alpha\in[0,1]$, the point $(x_\alpha,\kappa_{\alpha})$ is in the set $\mathcal{X}_\epsilon\times\mathbb R^n$, with:
    \begin{align}
        x_\alpha &=\alpha x_b+(1-\alpha)x_a,\text{  and  }\kappa_{\alpha}=\alpha \kappa_{b}+(1-\alpha)\kappa_{a}.
    \end{align}
    We distinguish four cases:
    \begin{enumerate}
        \item $(x_a,\kappa_{a})$ and $(x_b,\kappa_{b})$ are not in $\mathcal{S}$. Trivially, we have that:
        \begin{equation}
            \|\mathrm{Proj}(x_b,\kappa_{b})-\mathrm{Proj}(x_a,\kappa_{a})\|=\|\kappa_{b}-\kappa_{a}\|.
        \end{equation}
        \item For all  $\alpha\in [0, 1]$, $(x_\alpha, \kappa_{\alpha})$ lies in $\mathcal{S}$. Then, using the Mean Value Theorem, we get:
        \begin{align}
            \|\mathrm{Proj}(x_b,\kappa_{b})&-\mathrm{Proj}(x_a,\kappa_{a})\|   \nonumber\\
                {} &\le k_\mathcal{C}[\|x_b-x_a\|+\|\kappa_{b}-\kappa_{a}\|],
        \end{align}
        where $k_\mathcal{C}$ is given by (\ref{eq:jacobianKappaSmooth}).
        \item  When $(x_a, \kappa_{a})$ belongs to $\mathcal{S}$ but $(x_b, \kappa_{b})$ does not. Then, we define $\alpha^*$ by:
        \begin{equation}\label{eq:alphaStar}
            \alpha^*=\min_{\begin{array}{c}
                 0\le \alpha\le 1\\
                  (x_\alpha, \kappa_{\alpha})\notin \mathcal{S}
            \end{array}} \alpha.
        \end{equation}
        We have that, for all $\alpha\in[0,\alpha^*[$, $(x_\alpha, \kappa_{\alpha})$ is in $\mathcal{S}$. Then, using (\ref{eq:jacobianKappaSmooth}), we have that
        \begin{align}
            \|\mathrm{Proj}(x_{\alpha^*},&\kappa_{\alpha^*})-\mathrm{Proj}(x_a,\kappa_{a})\|   \nonumber\\
                {} &\le k_\mathcal{C}[\|x_b-x_a\|+\|\kappa_{b}-\kappa_{a}\|],
        \end{align}
        and, we also have
        \begin{align}
            \|\mathrm{Proj}(x_b,\kappa_{b})&-\mathrm{Proj}(x_{\alpha^*},\kappa_{\alpha^*})\|   \nonumber\\
                {} &= \|\kappa_{b}-\kappa_{\alpha^*}\|\le\|\kappa_{b}-\kappa_{a}\|,
        \end{align}
        therefore,
        \begin{align}
            \|\mathrm{Proj}(x_b,&\kappa_{b})-\mathrm{Proj}(x_a,\kappa_{a})\|   \nonumber\\
                {} &\le (k_\mathcal{C}+1)[\|x_b-x_a\|+\|\kappa_{b}-\kappa_{a}\|].
        \end{align}
        \item Finally, when both $(x_a,\kappa_{a})$ and $(x_b,\kappa_{b})$ are in $\mathcal{S}$ , but there are some $\alpha\in]0,1[$ for which $(x_\alpha, \kappa_{\alpha})$ is not in $\mathcal{S}$, we define $\alpha^*$ as in (\ref{eq:alphaStar}) and let
        \begin{equation}
            \beta^*=\max_{\begin{array}{c}
                 0\le \beta\le 1\\
                  (x_\beta, \kappa_{\beta})\notin \mathcal{S}
            \end{array}} \beta.
        \end{equation}
        We have that, for all $\alpha\in[0,\alpha^*[\cup]\beta^*,1]$, $(x_\alpha, \kappa_{\alpha})$ is in $\mathcal{S}$. Then, using (\ref{eq:jacobianKappaSmooth}), we have that
        \begin{align}
            \|\mathrm{Proj}(x_b,&\kappa_{b})-\mathrm{Proj}(x_{\beta^*},\kappa_{\beta^*})\|   \nonumber\\
            +\|\mathrm{Proj}(x_{\alpha^*},&\kappa_{\alpha^*})-\mathrm{Proj}(x_a,\kappa_{a})\|   \nonumber\\
                {} &\le 2k_\mathcal{C}[\|x_b-x_a\|+\|\kappa_{b}-\kappa_{a}\|],
        \end{align}
        and,
        \begin{equation}
            \mathrm{Proj}(x_{\beta^*},\kappa_{\beta^*})-\mathrm{Proj}(x_{\alpha^*},\kappa_{\alpha^*})= \kappa_{\beta^*}-\kappa_{\alpha^*}.
        \end{equation}
    \end{enumerate}
    Eventually, we can conclude that the projection (\ref{eq:projection}) is locally Lipschitz-continuous. Therefore, the smoothed control law (\ref{eq:smoothControl}) is also locally Lipschitz-continuous.
\end{proof}

\subsection{Safety and Stability Analysis}
To ensure the safety of the robot, we must guarantee that all trajectories starting at $x\left(0\right)\in\mathcal{X}_{\epsilon}$ will remain in $\mathcal{X}_{\epsilon}$ for all times. This is equivalent to proving that the set $\mathcal{X}_{\epsilon}$ is a positively invariant set for the dynamical system (\ref{eq:dynamicalSystem}). This is the result of the following theorem.
\begin{theorem}\label{theorem:safety}
    Consider the set $\mathcal{X}\subset\mathbb R^n$ that describes the free space and satisfies Assumption \ref{assumption:smoothBoundaries}. Consider the set $\mathcal{X}_\epsilon\in\mathbb R^n$ that describes the practical free space and is given by (\ref{eq:practicalfreespace}). Consider the closed-loop system (\ref{eq:closedLoopsystem}) under the locally Lipschitz-continuous control law (\ref{eq:smoothControl})-(\ref{eq:phi}). Then, the closed-loop system admits a unique solution and the set $\mathcal{X}_\epsilon$ is positively invariant.
\end{theorem}
\begin{proof}
    To prove the forward invariance of the set $\mathcal{X}_{\epsilon}$, we can verify that when $x\in\partial\mathcal{X}_{\epsilon}$, or equivalently, when $\mathbf{b}_{\complement\mathcal{X}_{\epsilon}}\left(x\right)$, and $\kappa_0\left(x\right)^\top\nabla\mathbf{b}_{\complement	\mathcal{X}}\left(x\right)\le 0$ we have

\[\setlength{\arraycolsep}{10pt}
  \begin{array}{ l l l}
    \left. \dot{x} \right|_{\mathbf{b}_{\complement	\mathcal{X}}\left(x\right)\ = \epsilon}&=&\left. \kappa(x) \right|_{\mathbf{b}_{\complement	\mathcal{X}}\left(x\right)\ = \epsilon}\\
    &=&\left. \hat{\Pi}\left(x\right)\kappa_0\left(x\right) \right|_{\mathbf{b}_{\complement	\mathcal{X}}\left(x\right)\ = \epsilon}\\
    &=&\Pi\left(x\right)\kappa_0\left(x\right)\\
    &=&\left(\mathbf{I}_n-\nabla\mathbf{b}_{\complement\mathcal{X}}\nabla\mathbf{b}_{\complement\mathcal{X}}^\top\right)\kappa_0\left(x\right)\\
    \end{array} 
\]
We multiply both sides by $\nabla\mathbf{b}_{\complement\mathcal{X}}^\top$ and we find that
\begin{equation}
       \nabla\mathbf{b}_{\complement\mathcal{X}}^\top\dot{x}=0\hspace{2mm}\text{when}\hspace{2mm} x\in\partial\mathcal{X}_{\epsilon},\hspace{2mm} \text{and}\hspace{2mm} \kappa_0\left(x\right)^\top\nabla\mathbf{b}_{\complement	\mathcal{X}}\left(x\right)\le 0.
\end{equation}

Thus, the trajectories will stay inside or at the boundary of the practical free space $\mathcal{X}_{\epsilon}$. Eventually, the set $\mathcal{X}_{\epsilon}$ is positively invariant.
Also, it follows from \cite[theorem 3.3]{khalil2002nonlinear} that the closed-loop system admits a unique solution.
\end{proof}
 
Theorem \ref{theorem:safety} states that safe navigation inside the practical free space $\mathcal{X}_\epsilon$ is guaranteed regardless of the chosen nominal controller $\kappa_0(x)$. Moreover, safety is ensured for any shape of the obstacles (convex or non-convex). The only mild requirement on the free space $\mathcal{X}$ is given in Assumption \ref{assumption:smoothBoundaries}.

Next, we consider the motion-to-goal feature, \textit{i.e.,} convergence of the robot's trajectories to the desired position $x_d$. The choice of the nominal controller $\kappa_0(x)$ might affect the convergence to the goal. The following result is based on choosing the traditional nominal controller
\begin{equation}\label{eq:nominalControl}
    \kappa_0(x)=-k(x-x_d), \hspace{5mm} k>0.
\end{equation}
\begin{theorem}\label{theorem:convergence}
    Consider the set $\mathcal{X}\subset\mathbb R^n$ that describes the free space and satisfies Assumption \ref{assumption:smoothBoundaries}. Consider the set $\mathcal{X}_\epsilon\in\mathbb R^n$ that describes the practical free space and is given by (\ref{eq:practicalfreespace}). Consider the closed-loop system (\ref{eq:closedLoopsystem}) under the locally Lipschitz-continuous control law (\ref{eq:smoothControl}), with $\kappa_0(.)$ as in (\ref{eq:nominalControl}). Then, the distance $||x-x_d||$ is non-increasing, the equilibrium point $x=x_d$ is exponentially stable, and trajectories converge to the set $\mathcal{E}\cup\{x_d\}$, where
    \begin{equation}\label{eq:E}
    \mathcal{E}:=\{x:\mathbf{b}_{\complement	\mathcal{X}}(x)=\epsilon,\left(x-x_d\right)=\lambda\nabla\mathbf{b}_{\complement	\mathcal{X}}\left(x\right), \lambda\in\mathbb{R}_{>0}\}
\end{equation}
is a set of measure zero.
\end{theorem}

\begin{proof}
    We consider the following positive definitive function
\begin{equation}
    V\left(x\right)=\frac{1}{2}||x-x_d||^2.
\end{equation}
Its time derivative along the trajectories of the closed-loop system \eqref{eq:closedLoopsystem} is given by
\begin{align}
&\dot{V}\left(x\right) = \left(x-x_d\right)^\top\kappa(x)\nonumber\\
&=\setlength{\arraycolsep}{2pt}
    \left\{ \begin{array}{ l l }
    -k||x-x_d||^2, &\mathbf{b}_{\complement	\mathcal{X}}\left(x\right)>\epsilon'\;\text{or} \\
    {}&  \kappa_0\left(x\right)^\top\nabla\mathbf{b}_{\complement	\mathcal{X}}\left(x\right)\ge 0,\\
    -k\left(x-x_d\right)^\top\hat{\Pi}\left(x\right)\left(x-x_d\right), &\mathbf{b}_{\complement	\mathcal{X}}\left(x\right)\le\epsilon' \;\text{and} \\
    {}&  \kappa_0\left(x\right)^\top\nabla\mathbf{b}_{\complement	\mathcal{X}}\left(x\right)\le 0.\\
  \end{array} \right. 
\end{align}
Let us prove that $\hat{\Pi}\left(x\right)$ is a positive semi-definite matrix. We have that 
\begin{align*}
    &\left(x-x_d\right)^\top\hat{\Pi}\left(x\right)\left(x-x_d\right) \\ &=\left(x-x_d\right)^\top\left(\mathbf{I}_n-\phi\left(x\right)\nabla\mathbf{b}_{\complement	\mathcal{X}}\left(x\right)\nabla\mathbf{b}_{\complement	\mathcal{X}}\left(x\right)^\top\right)\left(x-x_d\right)	\\
    &=||x-x_d||^2-\phi\left(x\right)||x-x_d||^2||\nabla\mathbf{b}_{\complement	\mathcal{X}}\left(x\right)||^2\cos^2{\theta}(x)\\
    &=||x-x_d||^2\left[1-\phi\left(x\right)\cos^2{\theta}(x)\right]\ge 0,
\end{align*}
\noindent since $0\le\phi\left(x\right)\le 1$, where $\theta(x)$ is the angle  between the two vectors $\left(x-x_d\right)$ and $\nabla\mathbf{b}_{\complement	\mathcal{X}}\left(x\right)$. Finally, we have that
\begin{equation}
    \dot{V}\left(x\right)\le 0,\hspace{5mm} \forall x\in\mathcal{X}_{\epsilon}.
\end{equation}
The points for which $\dot{V}\left(x\right)=0$ are either $x=x_d$ or are the points $x$ that satisfy $\mathbf{b}_{\complement	\mathcal{X}}\left(x\right)\le\epsilon'$ and $\kappa_0\left(x\right)^\top\nabla\mathbf{b}_{\complement	\mathcal{X}}\left(x\right)\le 0$ and $\left[1-\phi\left(x\right)\cos^2{\theta}(x)\right]=0$. The later implies that $\phi\left(x\right)=1$ and $\cos^2{\theta}(x)=1$. In other terms, the points $x$ lies on the boundary $\partial\mathcal{X}_{\epsilon}$ and satisfy $\left(x-x_d\right)=\lambda\nabla\mathbf{b}_{\complement	\mathcal{X}}\left(x\right)$ with $\lambda>0$. Finally, the solutions converge to the set of points $\{x_d\}\cup\mathcal{E}$, where $\mathcal{E}$ is defined in \eqref{eq:E}.
Since $\partial\mathcal{X}_\epsilon$ has measure zero, and $\mathcal{E}$ is a subset of $\partial\mathcal{X}_\epsilon$, it follows that $\mathcal{E}$ has also measure zero.

We can study the local behavior of the robot's trajectory in the neighborhood of the desired equilibrium. Since $x_d\in\mathbf{int}\left(\mathcal{X}_{\epsilon}\right)$, there exist a set represented by the ball $\mathcal{B}\left(x_d,r\right)$, where $r>0$, such that $\mathcal{B}\left(x_d,r\right)\subset\mathbf{int}\left(\mathcal{X}_{\epsilon}\right)$. The local dynamics of the robot on this set is given by $\dot{x}=-k(x-x_d)$, this implies that the desired equilibrium $x=x_d$ is exponentially stable.
\end{proof}

In view of Theorem \ref{theorem:convergence}, we can state that the solutions of the closed-loop system converge to either the desired equilibrium $x_d$ or to the set of measure zero $\mathcal{E}$. An undesired equilibrium $x\in\mathcal{E}$, by construction, must satisfy $(x-x_d)=\lambda\nabla\mathbf{b}_{\complement	\mathcal{X}}\left(x\right), \lambda>0$. In other terms, the points $x$, $\mathbf{P}_{\partial\mathcal{X}}(x)$ and $x_d$ are all collinear. This is only possible when the segment, whose endpoints are the desired equilibrium $x_d$ and the undesired equilibrium $x$, intersects the obstacle set. The invariance properties of the equilibria $\mathcal{E}$ depend on the topology of the free space.

\section{Convex Sphere Worlds}\label{section:ConvexSphereWorlds}
\subsection{Topology of the Obstacle Set}
\begin{figure}[t]
    \centering
    \includegraphics[width=0.8\columnwidth]{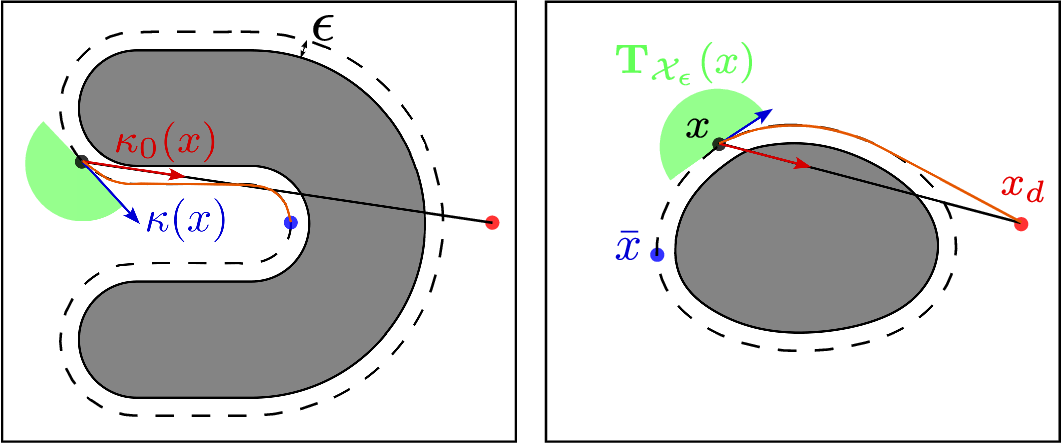}
    \caption{Different obstacle's topology affecting the nature of the equilibrium point:(left) a non-convex obstacle for which the trajectory of the robot converges to the undesired equilibrium, and (right) a convex obstacle for which the trajectory converges to the desired goal $x_d$.} 
    \label{fig:ConvexNonConvex}
\end{figure}

The nature of the undesired equilibria defined by the set $\mathcal{E}$ is directly related to the topology of the obstacles. For instance, for non-convex obstacles, and for a given goal $x_d$ for which the undesired equilibrium is located in the concave part, the trajectory of the robot may converge to undesired equilibrium. Therefore, the convexity of the obstacles is required. Also, besides the convexity, the flatness of the obstacles can affect the nature of the equilibria. When the undesired equilibrium point is located in a strongly curved part of the obstacle, it becomes unstable, and the more flat is the obstacle, the more stable is the undesired equilibrium point. Therefore, we consider the following assumption on the curvarute of the obstacles \cite[Assumption 2]{arslan2019sensor}

\begin{figure}[t]
    \centering
    \includegraphics[width=0.7\columnwidth]{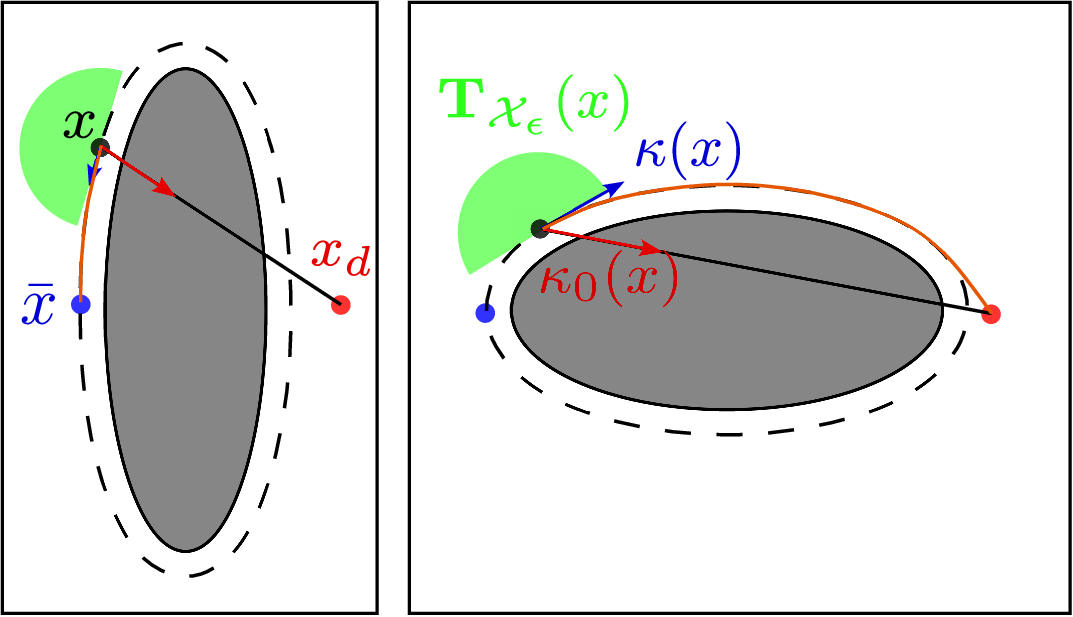}
    \caption{Two convex obstacles where the curvature affects the nature of the equilibrium point:(left) a flat obstacle, as viewed from the position of the vehicle, for which its trajectory converges to the undesired equilibrium, and (right) a strongly convex obstacle, as viewed from the position of the vehicle, for which its trajectory converges to the desired goal $x_d$.}
     
    \label{fig:strongConvexity}
\end{figure}
\begin{assumption}\label{assumption:curvatureCondition}
The Jacobian matrix $\mathbf{J}_{\mathbf{P}_{\partial\mathcal{X}}}\left(x\right)$ of the metric projection of any stationary
point $x\in\mathcal{E}$  onto the boundary $\partial\mathcal{X}$ of the free-space satisfy
\begin{equation}
    \begin{array}{l l}
        \mathbf{J}_{\mathbf{P}_{\partial\mathcal{X}}}\left(x\right)\prec\frac{||x_d-\mathbf{P}_{\partial\mathcal{X}}\left(x\right)||}{\epsilon+||x_d-\mathbf{P}_{\partial\mathcal{X}}\left(x\right)||}\mathbf{I}_n, & \forall x\in\mathcal{X},\\
    \end{array}
\end{equation}
\end{assumption}

For Assumption \ref{assumption:curvatureCondition} to hold, the practical obstacle which is defined by dilating the obstacle by $\epsilon$ must be contained entirely in the ball $\mathcal{B}(x_d,||\bar x-x_d||)$, where $\bar x\in\mathcal{E}$. Figures \ref{fig:ConvexNonConvex}-\ref{fig:strongConvexity} depicts the different obstacle topologies as discussed above.

\begin{lemma}\label{lemma:symmetricJacobian}
   Consider the set $\mathcal{X}\subset\mathbb R^n$ that describes the free space and satisfies Assumption \ref{assumption:smoothBoundaries}. Consider the set $\mathcal{X}_\epsilon\in\mathbb R^n$ that describes the practical free space and is given by (\ref{eq:practicalfreespace}). Therefore, the Jacobian $\mathbf{J}_{\mathbf{P}_{\mathcal{O}_i}}(x)$ is a symmetric matrix, for all $x\in\mathcal{X}$.
\end{lemma}
\begin{proof}
    See \cite[Lemma 7]{arslan2019sensor}.
\end{proof}

We summarize the nature of the undesired equilibria and the desired goal in the following theorem: 
\begin{theorem}\label{theorem:unstableEquilibria}
    Consider the set $\mathcal{X}\subset\mathbb R^n$ that describes the free space and satisfies Assumption \ref{assumption:smoothBoundaries}. Consider the set $\mathcal{X}_\epsilon\in\mathbb R^n$ that describes the practical free space and is given by (\ref{eq:practicalfreespace}). Consider the closed-loop system (\ref{eq:closedLoopsystem}) under the locally Lipschitz-continuous control law (\ref{eq:smoothControl}), with $\kappa_0(.)$ as in (\ref{eq:nominalControl}). If Assumption \ref{assumption:curvatureCondition} holds, then 
    \begin{enumerate}
        \item all the undesired equilibria $\bar x\in\mathcal{E}$ are unstable, and
        \item the desired equilibrium $x_d$ is locally exponentially stable and almost globally asymptotically stable.
    \end{enumerate}
\end{theorem}

\begin{proof}
To prove {\it item 1)}, first, we consider the ball $\mathcal{B}(\bar x,r)$, where $\bar x\in\mathcal{E}$. We define the set $\mathcal{P}=\{x:\left(x-x_d\right)^\top\nabla\mathbf{b}_{\complement\mathcal{X}}\left(x\right)\ge 0\}$ . 
    The local dynamics when the configurations of the robot are restricted to the set $\mathcal{B}\cap\mathcal{P}\cap\partial\mathcal{X}_{\epsilon}$ are given by

\begin{equation}\label{eq:localController}
    u=-k\left[x-x_d+g(x)(x-\mathbf{P}_{\partial\mathcal{X}}(x))\right],
\end{equation}
where
\begin{equation}
    g\left(x\right)=\frac{\left(x_d-\mathbf{P}_{\partial\mathcal{X}}(x)\right)^\top\left(x-\mathbf{P}_{\partial\mathcal{X}}(x)\right)}{\epsilon^2}-1.
\end{equation}

The Jacobian of the controller $u$ is given by
\begin{equation}
    \textbf{J}_u(x)=-k[\textbf{I}_n+(x-\mathbf{P}_{\partial\mathcal{X}}(x)\textbf{J}_g(x)+g(x)(\textbf{I}_n-\textbf{J}_{\mathbf{P}_{\partial\mathcal{X}}}(x))],
\end{equation}
where the Jacobian of $g(x)$ is
\begin{align}
    \textbf{J}_g(x) &=
         \frac{(x-\mathbf{P}_{\partial\mathcal{X}}(x))^\top(-\textbf{J}_{\mathbf{P}_{\partial\mathcal{X}}}(x))}{\epsilon^2} \nonumber\\
       &+\frac{(x_d-\mathbf{P}_{\partial\mathcal{X}}(x))^\top(\textbf{I}_n-\textbf{J}_{\mathbf{P}_{\partial\mathcal{X}}}(x))}{\epsilon^2}.
\end{align}
We have from \cite[Proposition 3.7]{Fitzpatrick1982} that $\textbf{J}_{\mathbf{P}_{\partial\mathcal{X}}}(x)(x-\mathbf{P}_{\partial\mathcal{X}}(x))=0$. Therefore, from Lemma \ref{lemma:symmetricJacobian},  we have that $(x-\mathbf{P}_{\partial\mathcal{X}}(x))^\top\textbf{J}_{\mathbf{P}_{\partial\mathcal{X}}}(x)=0$. It follows that the Jacobian of $g$ reduces to:
\begin{equation}
    \textbf{J}_g(x) =\frac{(x_d-\mathbf{P}_{\partial\mathcal{X}}(x))^\top(\textbf{I}_n-\textbf{J}_{\mathbf{P}_{\partial\mathcal{X}}}(x))}{\epsilon^2}.
\end{equation}
At a given undesired equilibrium point $\bar x\in\mathcal{E}$, one has
\begin{equation}
    \textbf{J}_g(\bar x) =\frac{(x_d-\mathbf{P}_{\partial\mathcal{X}}(\bar x))^\top(\textbf{I}_n-\textbf{J}_{\mathbf{P}_{\partial\mathcal{X}}}(\bar x))}{\epsilon^2}.
\end{equation}
Since $(x_d-\mathbf{P}_{\partial\mathcal{X}}(\bar x))$ and $(\bar x-\mathbf{P}_{\partial\mathcal{X}}(\bar x))$ are two collinear vectors, then one has $(x_d-\mathbf{P}_{\partial\mathcal{X}}(\bar x))^\top\textbf{J}_{\mathbf{P}_{\partial\mathcal{X}}}(\bar x)=0$ and
\begin{equation}
    \textbf{J}_g(\bar x) =\frac{(x_d-\mathbf{P}_{\partial\mathcal{X}}(\bar x))^\top}{\epsilon^2}.
\end{equation}
The Jacobian of the control law $u$ at $\bar x$:
\begin{equation}
    \begin{split}
\textbf{J}_u(\bar x)=& -k[\textbf{I}_n+\frac{(\bar x-\mathbf{P}_{\partial\mathcal{X}}(\bar x))(x_d-\mathbf{P}_{\partial\mathcal{X}}(\bar x))^\top}{||\bar x-\mathbf{P}_{\partial\mathcal{X}}(\bar x)||^2} \\
 &+g(\bar x)(\textbf{I}_n-\textbf{J}_{\mathbf{P}_{\partial\mathcal{X}}}(\bar x))].
\end{split}
\end{equation}
We can write
\begin{equation}
    (x_d-\mathbf{P}_{\partial\mathcal{X}}(\bar x))=-||x_d-\mathbf{P}_{\partial\mathcal{X}}(\bar x)||\frac{(\bar x-\mathbf{P}_{\partial\mathcal{X}}(\bar x))}{||\bar x-\mathbf{P}_{\partial\mathcal{X}}(\bar x)||},
\end{equation}
It follows that
\begin{equation}
    \begin{split}
\textbf{J}_u(\bar x)=& -k[\textbf{I}_n-\textbf{A}\frac{||x_d-\mathbf{P}_{\partial\mathcal{X}}(\bar x)||}{||\bar x-\mathbf{P}_{\partial\mathcal{X}}(\bar x)||}\\
 &  +g(\bar x)(\textbf{I}_n-\textbf{J}_{\mathbf{P}_{\partial\mathcal{X}}}(\bar x))],
\end{split}
\end{equation}
where we have defined the matrix $\textbf{A}$ as
\begin{equation}
    \textbf{A}:=\frac{(\bar x-\mathbf{P}_{\partial\mathcal{X}}(\bar x))(\bar x-\mathbf{P}_{\partial\mathcal{X}}(\bar x))^\top}{||\bar x-\mathbf{P}_{\partial\mathcal{X}}(\bar x)||^2},
\end{equation}

\begin{equation}
\textbf{J}_u(\bar x)= -k[\textbf{I}_n-\textbf{A}\frac{||x_d-\mathbf{P}_{\partial\mathcal{X}}(\bar x)||}{||\bar x-\mathbf{P}_{\partial\mathcal{X}}(\bar x)||}+g(\bar x)(\textbf{I}_n-\textbf{J}_{\mathbf{P}_{\partial\mathcal{X}}}(\bar x))].
\end{equation}

We have
\begin{equation}
    g(\bar x)=-\frac{||x_d-\mathbf{P}_{\partial\mathcal{X}}(\bar x)||}{\epsilon}-1<-2,
\end{equation}

then
\begin{align}
 \textbf{J}_u(\bar x))   
     &=-kg(\bar x)[\frac{||x_d-\mathbf{P}_{\partial\mathcal{X}}(\bar x)||}{\epsilon+||x_d-\mathbf{P}_{\partial\mathcal{X}}(\bar x)||}(\textbf{I}_n+\textbf{A})-\textbf{J}_{\mathbf{P}_{\partial\mathcal{X}}}(\bar x)].
\end{align}

If Assumption \ref{assumption:curvatureCondition} holds, the Jacobian of the controller $u$ satisfies
\begin{equation}\label{eq:JacobianController}
    \textbf{J}_u(\bar x) \succ -kg(\bar x)\left(\frac{||x_d-\mathbf{P}_{\partial\mathcal{X}}(\bar x)||}{\epsilon+||x_d-\mathbf{P}_{\partial\mathcal{X}}(\bar x)||}\textbf{A}\right).
\end{equation}
Finally, the Jacobian of the controller $u$ evaluated at $\bar x$ has at least one strictly positive eigenvalue. Thus, according to \cite[Theorem 3.2]{khalil2015nonlinear}, all points $\bar x\in\mathcal{E}$ are unstable.

For {\it item 2)}, we prove that the basin of attraction of the undesired equilibria is a set of measure zero. Firstly, we denote by $\phi_t$ the flow of the closed-loop dynamical system (\ref{eq:closedLoopsystem}), and the stable manifold $\mathcal{S}$ for each undesired equilibrium point which satisfies
    \begin{equation}
        \lim_{t\to\infty}\phi_t(c)=\bar x,\hspace{4mm} \forall c\in\mathcal{S},\hspace{4mm}\text{where }\bar x\in\mathcal{E}. 
    \end{equation}
    The Jacobian of the controller $u$ evaluated at a point $\bar x\in\mathcal{E}$ satisfies (\ref{eq:JacobianController}), therefore, the Jacobian, as said in the proof of Theorem \ref{theorem:unstableEquilibria}, has at least one positive eigenvalue. Hence, the stable manifold $\mathcal{S}$ is at most $(n-1)$-dimensional manifold \cite[The Stable Manifold Theorem, Pg 107]{Perko2001}, and as a result, it is measure zero in the $n$-dimensional space. Since the closed-loop system (\ref{eq:closedLoopsystem}) admits unique solutions, the global stable manifold at $\bar x\in\mathcal{E}$, defined as \cite[Definition 3, Pg 113]{Perko2001}
    \begin{equation}
        W^s(\bar x)=\bigcup_{t\le0}\phi_t(\mathcal{S}),
    \end{equation}
    is also a measure zero set.
\end{proof}

\section{Numerical Simulation}\label{section:NumericalSimulation}
\begin{figure}[t]
    \centering
\includegraphics[width=0.4\columnwidth]{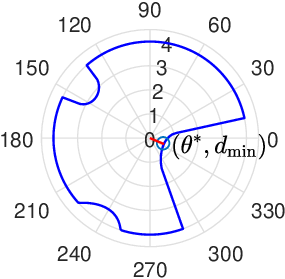}
\includegraphics[width=0.8\columnwidth]{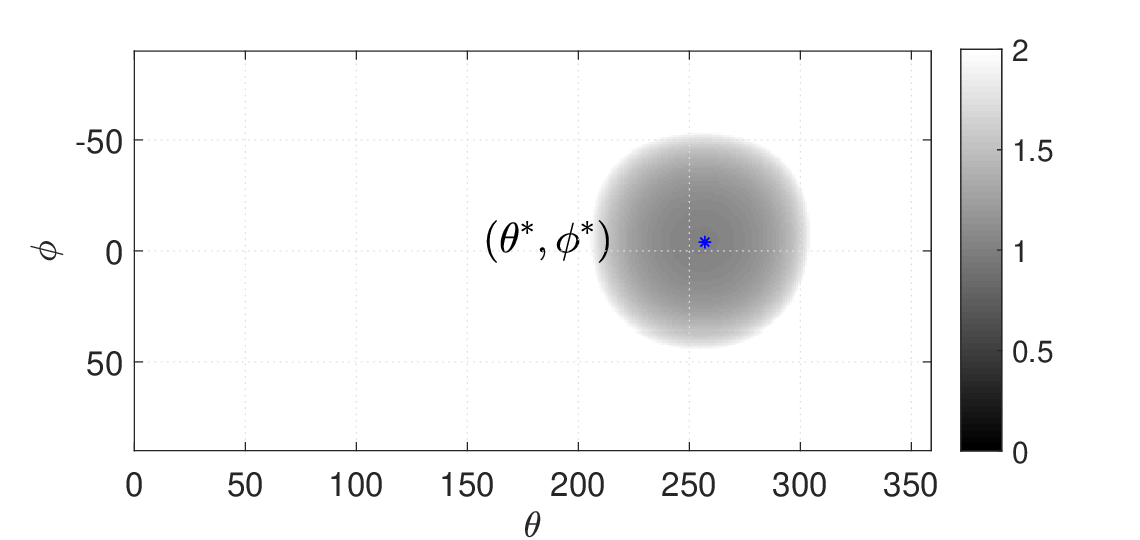}
    \caption{(Top) Example of a 2D LIDAR reading, which represents a polar curve. (Bottom) A presentation of the 3D LIDAR reading in a 2D grayscale map with the angles $\theta$ and $\phi$ as the axis and the color of a point $(\theta,\phi)$ is attributed according to the value of $\rho$, (black) when  $\rho=0$, (white) when $\rho=R_s$ and shades of (gray) in between.} 
    \label{fig:sensor}
\end{figure}
\begin{figure}[t]
    \centering
    \includegraphics[width=1\columnwidth]{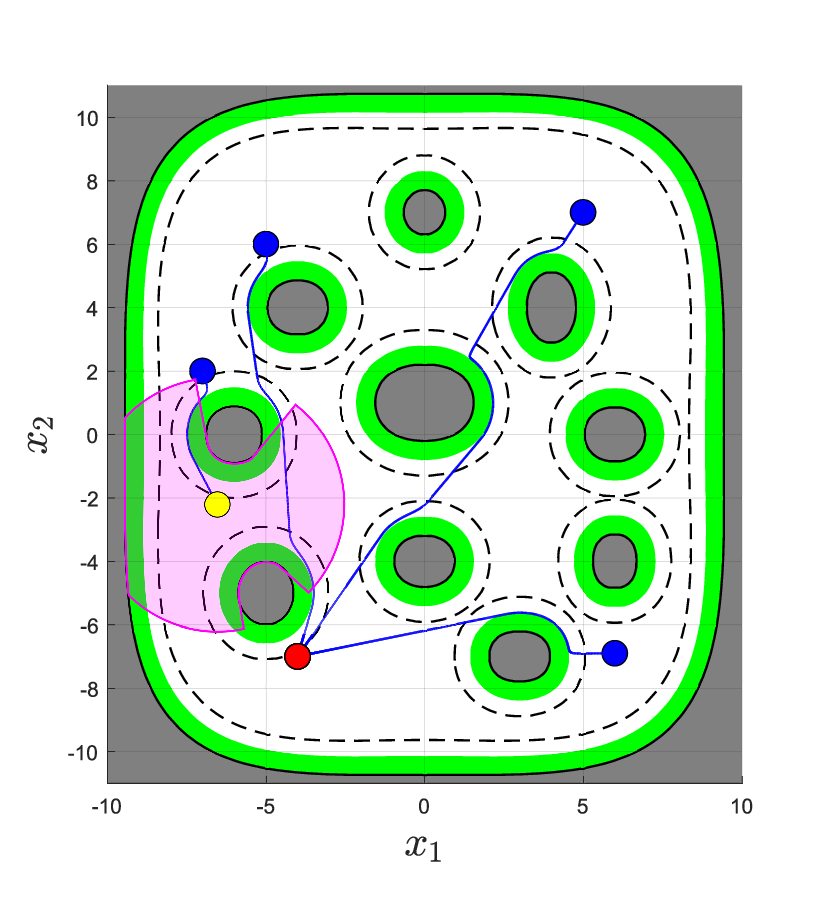}
    \caption{the resulting navigation trajectories, in a 2D environment, of the smooth control law (\ref{eq:smoothControl}) starting at a set of initial positions (blue) away from the goal (red) while avoiding the obstacles (gray). The (green) region around the obstacle (gray) is a dilation for the latest by the parameter $\epsilon$ while the (black) dashed line is a dilation by the parameter $\epsilon'$. The (magenta) area represents the sensor range for the actual position of the robot (yellow).} 
    \label{fig:Case2D}
\end{figure}

\begin{figure}[t]
    \centering
     \includegraphics[trim={9cm 0cm 9cm 0cm},clip,width=1\columnwidth]{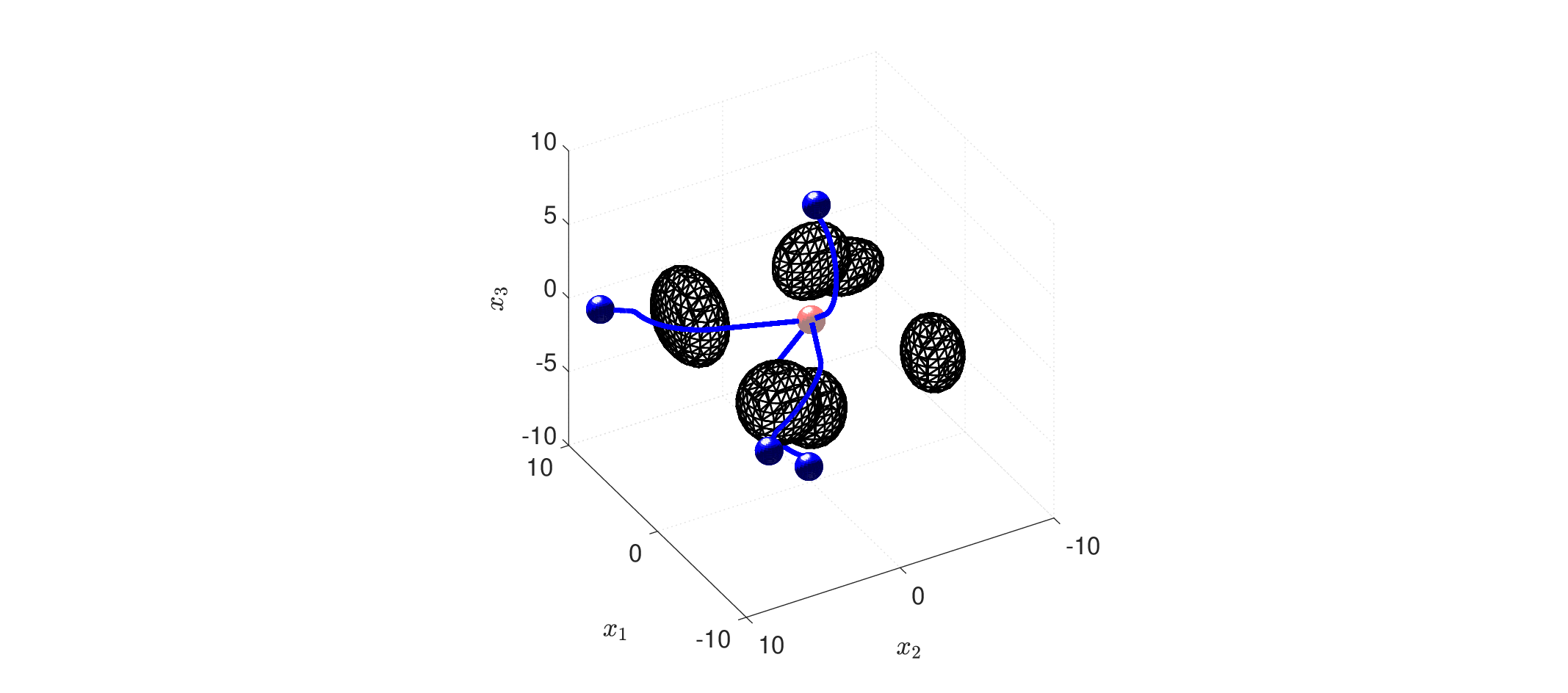}
    \caption{the resulting navigation trajectories, in a 3D environment, of the smooth control law (\ref{eq:smoothControl}) starting at a set of initial positions (blue) away from the goal (red) while avoiding the obstacles (gray).} 
    \label{fig:Case3D}
\end{figure}
In order to demonstrate the robot's safe navigation in an unknown environment, we run some 2D/3D numerical simulations that aim to visualize the ability of the robot to avoid obstacles while moving towards the goal under our smooth controller (\ref{eq:smoothControl}). For the 2D case,  we define the free space $\mathcal{X}_1\subset\mathbb R^2$ as follows,
\begin{align}
    \mathcal{X}_1:=\{x=(x_1,x_2)\in\mathbb R^2 : &f_i(x)\le 0, i\in\{1,..,M\}, \nonumber\\
    &\text{and}\hspace{2mm}g(x)\ge 0\},
\end{align}
where $g(x)=q x_1^{2n} + p x_2^{2n} - D^{2n}$ and$f_i(x)=a_i(x_1-x_{0,i})^2 + b_i(x_2-y_{0,i})^2 + c_i(x_1-x_{0,i})^4+d_i(x_2-y_{0,i})^4-R_{0,i}$. The parameters $x_{0,i},y_{0,i}\in \mathbb R$, $q,p,D,a_i,b_i,c_i,d_i,R_{0,i}\in\mathbb R_{>0}$ and $n\in\mathbb N\setminus\{0\}$ are the obstacles characteristics. To measure the distance of the robot relative to the obstacles, we use a 2D LIDAR range sensor with a limited range $R_s$, which we simulate using a function that returns the polar curve $\rho(\theta;x)$, where $\rho\in[0,R_s]$ and $\theta\in[0,2\pi)$ are the radial distance and the polar angle respectively. To measure the distance between the position $x$ of the robot and the obstacle region, we calculate the minimum value of $\rho$ with respect to $\theta$, and we use the corresponding angle $\theta^*$ to evaluate the vector $\nabla\mathbf{b}_{\complement\mathcal{X}_1}(x)$, such that:
\begin{equation}
    \nabla\mathbf{b}_{\complement\mathcal{X}_1}(x)=(-\cos(\theta^*),-\sin(\theta^*)).
\end{equation}
For this simulation, we take the desired goal at $x_d=(-4,-7)$, the robots radius $R=0.4$, the controller parameters $\epsilon=0.6,\epsilon'=1.1,k=0.5$, and the 2D sensor range $R_s=4$. For the 3D case, we define the free space $\mathcal{X}_2\subset\mathbb R^3$ as follows,
\begin{equation}
    \mathcal{X}_2:=\{x=(x_1,x_2,x_3)\in\mathbb R^3: h_i(x)\le 0,i\in\{1,..,M\}\},
\end{equation}
where $h_i(x)=a_i(x_1-x_{0,i})^2 + b_i(x_2-y_{0,i})^2 + c_i (x_3-z_{0,i}).^2 - R_{0,i}^2$. The parameters $x_{0,i},y_{0,i},z_{0,i}\in \mathbb R$, $a_i,b_i,c_i,R_{0,i}\in\mathbb R_{>0}$ are the obstacles characteristics. We use a function that simulates a 3D LiDAR range sensor and returns a surface defined in the spherical coordinates by the equation $\rho(\theta,\phi;x)$, where $\rho\in[0,R_s]$ is the radial distance, $\theta\in[0,2\pi)$ is the polar angle and $\phi\in[-\pi/2,\pi/2]$ is the azimuthal angle. The vector $\nabla\mathbf{b}_{\complement\mathcal{X}_2}(x)$ is given by
\begin{equation}
     \nabla\mathbf{b}_{\complement\mathcal{X}_2}(x)=-(\cos(\theta^*)\cos(\phi^*),\sin(\theta^*)\cos(\phi^*),\sin(\phi^*)),
\end{equation}
where $\theta^*$ and $\phi^*$ are the angles that corresponds to minimum of $\rho$. For the 3D case, we take the desired goal at $x_d=(0,0,1)$, the robots radius $R=0.8$, the controller parameters $\epsilon=1,\epsilon'=1.4,k=0.5$, and the 3D sensor range $R_s=2$. The figures \ref{fig:Case2D} and \ref{fig:Case3D} illustrate the resulting navigation trajectories for different initial conditions for the robot in the 2D and 3D environments.

\section{Conclusion}\label{section:conclusion}
In this paper, we proposed a sensor-based feedback controller that solves the safe autonomous navigation problem in $n$-dimensional unknown environments. Our controller stabilizes the robot using the nominal control law and switches to avoidance when it comes closer to the obstacles but the transition between the two modes is smooth. For obstacles that satisfy the strong convexity assumption (Assumption \ref{assumption:curvatureCondition}), our controller guarantees almost global asymptotic stability and safe navigation. The fact that our feedback controller uses only range and bearing to the nearest obstacle makes it very suitable for practical real-time implementation using range sensors. Considering robots with higher-order dynamics is an interesting future extension of this work.

\bibliographystyle{ieeetr}
\bibliography{references}
\end{document}